\documentclass[twoside]{article}
\usepackage{fullpage}
\usepackage{amssymb,amsmath,amsthm,wrapfig}
\usepackage{wasysym}
\usepackage[inline]{enumitem}
\usepackage{anyfontsize}
\usepackage{hyperref}
\usepackage{authblk}
\usepackage{subcaption}
\usepackage{timetravel}

\newtheorem{conjecture}{Conjecture}

\newtheorem{theorem}{Theorem}
\newtheorem{lemma}{Lemma}
\theoremstyle{definition}
\newtheorem{proofWithFormula}{Proof}

\newcommand{\down}{\bigtriangledown}
\newcommand{\up}{\bigtriangleup}
\newcommand{\both}{{\normalfont{\davidsstar}}}
\newcommand{\RC}{Q} 
\newcommand{\rc}{q} 
\newcommand{\etal}{{et~al.~}}

\newcommand{\comp}{{\rm comp}}
\newcommand{\odd}{{\rm odd}}

\usepackage[utf8]{inputenc}
\usepackage{graphicx}
\graphicspath{{fig/}}

\usepackage[normalem]{ulem}
\newcommand{\para}[1]{\smallskip\noindent{\bf #1}}

\makeatletter
\newcommand{\newreptheorem}[2]{\newtheorem*{rep@#1}{\rep@title}\newenvironment{rep#1}[1]{\def\rep@title{#2 \ref*{##1}}\begin{rep@#1}}{\end{rep@#1}}}
\makeatother
\theoremstyle{plain}
\newreptheorem{theorem}{Theorem}

\usepackage{fancyhdr}
\title{Maximum Matchings and Minimum Blocking Sets in $\Theta_6$-Graphs}

\author[1]{Therese Biedl}
\author[1]{Ahmad Biniaz}
\author[1]{Veronika Irvine}
\author[2]{Kshitij Jain}
\author[3]{Philipp Kindermann}
\author[1]{Anna Lubiw}

\affil[1]{David R.~Cheriton School of Computer Science, \authorcr University of Waterloo, Canada\authorcr
  \texttt{$\{$\href{biedl@uwaterloo.ca}{biedl},\href{abiniaz@uwaterloo.ca}{abiniaz},\href{virvine@uwaterloo.ca}{virvine},\href{alubiw@uwaterloo.ca}{alubiw}$\}$@uwaterloo.ca}}
\affil[2]{Borealis AI, Waterloo, Canada\authorcr
  \texttt{\url{kshitij.jain.1@uwaterloo.ca}}}
\affil[3]{Lehrstuhl für Informatik I, Universität Würzburg, Germany\authorcr
  \texttt{\url{philipp.kindermann@uni-wuerzburg.de}}}

\begin{document}

\maketitle
\begin{abstract}
$\Theta_6$-Graphs graphs are important geometric graphs that have many applications especially in wireless sensor networks.
They are equivalent to Delaunay graphs where empty equilateral triangles take the place of empty circles.
We investigate lower bounds on the size of maximum matchings in these graphs. 
The best known lower bound is $n/3$, where $n$ is the number of vertices of the graph.
Babu \etal (2014) conjectured that any $\Theta_6$-graph has a (near-)perfect matching 
(as is true for standard Delaunay graphs). Although this conjecture remains open, 
we improve the lower bound to $(3n-8)/7$. 

We also relate the size of maximum matchings in $\Theta_6$-graphs to the minimum size of
a \emph{blocking set}.  Every edge of a $\Theta_6$-graph on point set $P$ corresponds to
an empty triangle that contains the endpoints of the edge but no other point of $P$.  
A \emph{blocking set} has at least one point in each such triangle. 
We prove that the size of a maximum matching is at least $\beta(n)/2$ where $\beta(n)$ is 
the minimum, over all $\Theta_6$-graphs with $n$ vertices, of the minimum size of a blocking set.  	
In the other direction, lower bounds on matchings can be used to prove bounds on~$\beta$, 
allowing us to show that $\beta(n)\geq 3n/4-2$.
\end{abstract}

\section{Introduction}\label{sec:intro}

One of the many beautiful properties of Delaunay triangulations is that they always 
contain a (near-)perfect matching, that is, at most one vertex is unmatched,
as proved by Dillencourt~\cite{Dillencourt1990}.
This is one example of a structural property of a so-called \emph{proximity graph}. 
A proximity graph is determined by a set $\cal S$ of geometric objects in the plane, such as all disks, or all
axis-aligned squares.  Given such a set $\cal S$ and a finite point set $P$, 
we construct a proximity graph with vertex set $P$ and with an edge $(p,q)$ if there is an object from $\cal S$
that contains $p$ and~$q$ and no other point of $P$.  When $\cal S$ consists of all disks, then we get the Delaunay
triangulation. Proximity graphs are often defined in a more general way, with constraints on how the objects may
touch points~$p$ and $q$, but this narrow definition suffices for our purposes.

Various structural properties have been proved for different classes of proximity graphs. 
Another example, besides the (near-)perfect matching example above, 
is that the $L_\infty$-Delaunay graph, which is a proximity graph defined in terms of 
the set $\cal S$ of all axis-aligned squares, has the even stronger property of always 
having a Hamiltonian path~\cite{Abrego2009}. 

Our paper is about structural properties of $\Theta_6$-graphs, which are the 
proximity graphs determined by equilateral  triangles with a horizontal edge.  
More precisely, for any finite point set $P$, define $G^\up(P)$ to be the 
proximity graph of $P$ with respect to 
\emph{upward} equilateral triangles $\up$, define $G^\down(P)$ to be the 
proximity graph of $P$ with respect to \emph{downward} equilateral triangles $\down$, 
and define $G^\both(P)$, the $\Theta_6$-graph of $P$, to be their union.  
In particular, $G^\both(P)$ has an edge between points $p$ and $q$ if and only if 
there is an equilateral triangle with a horizontal side that contains $p$ and $q$ 
and no other point of $P$.  Such a triangle can be shrunk to an \emph{empty triangle} 
that has one of $p$ or $q$ at a corner, the other point on its boundary, and no 
points of $P$ in its interior. 

The graphs $G^\up(P)$ and $G^\down(P)$ are \emph{triangular-distance} (or ``TD'') 
Delaunay graphs, first introduced by Chew~\cite{Chew1989}.  
Clarkson~\cite{Clarkson1987} and Keil~\cite{Keil1988} first introduced 
$\Theta_6$-graphs(via a different definition), and the equivalence with the above 
definition was proved by Bonichon \etal\cite{Bonichon2010}. 
See Section~\ref{subsec:background} for more information.

We explore two conjectures about $\Theta_6$-graphs.  

\begin{conjecture} [Babu \etal\cite{Babu2014}]
\label{conj:perfect-matching}
    Every $\Theta_6$-graph has a (near-)perfect matching.
\end{conjecture}

\begin{figure}[t]
  \centering
	\subcaptionbox{\label{fig:perfect-matching-a}}{\includegraphics[page=1,scale=.8]{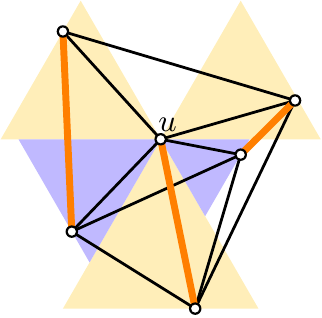}}
	\hfil
	\subcaptionbox{\label{fig:perfect-matching-b}}{\includegraphics[page=3,scale=.8]{example2}}
	\hfil
	\subcaptionbox{\label{fig:perfect-matching-c}}{\includegraphics[page=4,scale=.8]{example2}}
    
  \caption{A $\Theta_6$-graph on $n=6$ points with a perfect matching and a blocking set of size $5$.  
  (a)~A perfect matching. Empty triangles corresponding to edges of~$u$ are highlighted. 
  (b)~A blocking set $B$ of size $n-1$. Edges have the same color as their blocking point.  
  (c)~$G^\both(P \cup B)$.  For every edge, one endpoint is in~$B$.    }
  \label{fig:perfect-matching}
\end{figure}

See Figure~\ref{fig:perfect-matching} for an example.  
The best known bound is that every $\Theta_6$-graph on $n$ points has a matching of size at least~$n/3$ 
minus a small constant---in fact, this bound holds for any planar graph with minimum 
degree~3~\cite{NB79}, hence for any triangulation and in particular for each of $G^\up$ and $G^\down$
(modulo the small additive constant)---see Babu \etal\cite{Babu2014} for the exact bound of $\lceil (n-1)/3 \rceil$.
Our main result is an improvement of this lower bound:

\begin{theorem}
\label{thm:matching-bound}
    Every $\Theta_6$-graph on $n$ points has a matching of size $(3n-8)/7$. 
\end{theorem}

We prove Theorem~\ref{thm:matching-bound} in Section~\ref{sec:matching-bound} 
using the same technique that has been used for matchings in planar proximity graphs, 
namely the Tutte-Berge theorem, which relates the size of a maximum matching in a 
graph to the number of components of odd cardinality after removing some vertices.  
In our case, this approach is more complicated because $\Theta_6$-graphs are not planar. 

Our second main result relates the size of matchings to the size of \emph{blocking} 
or \emph{stabbing} sets of proximity graphs, which were introduced by 
Aronov \etal\cite{aronov2011witness} for purposes unrelated to matchings.
For a proximity graph $G(P)$ defined in terms of a set of objects~$\cal S$,
 we say that a set $B$ of points 
\emph{blocks} $G(P)$ if $B$ has a point in the interior of any object from $\cal S$ 
that contains exactly two points of $P$, i.e., the set $B$ destroys all the edges 
of $G(P)$, or equivalently, $G(P \cup B)$ has no edges between vertices in $P$;
see Figure~\ref{fig:perfect-matching}.
See Section~\ref{subsec:background} for previous results on blocking sets.

For a set of points $P$, let $\beta(P)$ be the minimum size of a blocking set 
of $G^\both(P)$. Let $\beta(n)$ be the minimum, over all point sets $P$ of size $n$, of $\beta(P)$. 
It is known that $\beta(n) \ge \lceil (n-1)/2\rceil$
since that is a lower bound for blocking all 
$G^\up$-graphs of $n$ points~\cite{Biniaz2015}. 
Let $\mu(n)$ be the minimum, over all point sets $P$ of size $n$, of 
the size of a maximum matching in 
$G^\both(P)$.  Conjecture~\ref{conj:perfect-matching} can hence be restated as $\mu(n)\geq \lceil(n-1)/2\rceil$.
We relate the parameters $\mu$ and $\beta$ as follows.

\wormhole{thmReductions}
\begin{theorem}
\label{thm:reductions}
(a) For any point set $P$ of $n$ points in the plane, $G^\both(P)$ has a matching of size $\beta(n)/2$, i.e., $\mu(n) \ge \beta(n)/2$.
(b) On the other hand, if $\mu(n)\geq cn+d$ for some constants $c,d$, then $\beta(n)\geq (cn+d)/(1-c)$.
\end{theorem}

The two statements in the theorem are proved in Section~\ref{sec:blocking}.
The idea of using bounds on blocking sets to obtain bounds on matchings is new, and is proved via the Tutte-Berge theorem. 
Theorem~\ref{thm:reductions} has two consequences.  The first is that 
Theorem~\ref{thm:matching-bound} implies that $\beta(n)\geq 3n/4-2$.
The second consequence is that Conjecture~\ref{conj:perfect-matching} is equivalent to the following:

\begin{conjecture}
  \label{conj:blocking}
  $\beta(n) \ge n-1$.
\end{conjecture}

In the remainder of the paper, 
we explore an approach to obtaining lower bounds on~$\beta(n)$.
For $B$ to be a blocking set, it must have a point in 
every empty triangle of~$P$ that defines an edge in $G^\both(P)$.
Let~$\alpha(n)$ be the maximum number of pairwise internally-disjoint 
empty triangles of any point set of size~$n$.
Clearly, $\beta(P) \ge \alpha(P)$ and $\beta(n) \ge \alpha(n)$.
Conjecture~\ref{conj:perfect-matching} would be proved if we could show that $\alpha(n) \ge n-1$.
However, we give an example of a point set~$P$ of size~$n$ with $\alpha(P) \le 3n/4$, which 
shows that $\alpha(n) \le 3n/4$. 
We also explore a previously-studied variant where the empty triangles
must be completely disjoint, i.e., even their boundaries must be disjoint. 
If $D$ is such a set, then every empty triangle in~$D$ corresponds to an edge 
in $G^\both(P)$, and these edges share no endpoint because the triangles are disjoint. 
Then $D$ corresponds to a \emph{strong matching} in $G^\both(P)$.
Strong matchings were introduced by \'Abrego \etal\cite{Abrego2004,Abrego2009}
for the case where the empty objects are line segments, rectangles, disks, or squares. They 
showed that  Delaunay and $L_\infty$-Delaunay graphs need not have strong 
(near-)perfect matchings (for disks and squares, respectively). 
See the following subsection for further background. 
Biniaz \etal\cite{Biniaz2015} proved that for any point set of size $n$, 
$G^\up(P)$ has a strong matching of at least $\lceil \frac{n-1}{9} \rceil$ edges and 
$G^\both(P)$ has a strong matching of at least $\lceil (n-1)/4 \rceil$ edges.
We prove an upper bound 
on the size of a strong matching in $\Theta_6$-graphs 
by giving an example where the maximum strong matching in $G^\both(P)$ 
has $2n/5$ edges.  

In the final Section~\ref{sec:properties}, we prove some additional bounds on 
the number of edges, maximum vertex degree, and maximum independent set of $\Theta_6$-graphs.

\subsection{Background}\label{subsec:background}

\para{$\Theta_6$-graphs and TD-Delaunay graphs.}  
The $\Theta_6$-graph on a set $P$ of points in the plane, as originally defined 
by Clarkson~\cite{Clarkson1987} and Keil~\cite{Keil1988}, is a geometric graph 
with vertex set $P$ and edges constructed as follows.  For every point $p \in P$, 
place 6 rays emanating from $p$ at angles that are multiples of $\pi /3$ radians 
from the positive $x$-axis.  These rays partition the plane into $6$ cones with 
apex $p$, which we label $C_1, \ldots, C_6$ in counterclockwise order starting 
from the positive $x$-axis; see Figure~\ref{cones-fig-a}.  
Add an edge from $p$ to the \emph{closest} point in each cone $C_i$,
where the distance between the apex $p$ and a point $q$ in $C_i$ is measured by 
the Euclidean distance from $p$ to the projection of $q$ on the bisector of $C_i$ 
as depicted in Figure~\ref{cones-fig-a}.
If the apex is not clear from the context, then we use $C_i^p$ to denote the 
cone $C_i$ with apex $p$. We sometimes refer to $C_i^p$ as the $i^\text{th}$ {\em cone} of $p$. 
It is straight-forward to show that this definition of $\Theta_6$-graphs is 
equivalent to the definition of $G^\both(P)$.
For any such edge, there is an equilateral up or down triangle with $p$ at one 
corner and $q$ on the opposite side, and no other points of $P$ inside.  
Thus, the edge is in $G^\both(P)$. In the other direction, if $e = (p,q)$ is an 
edge of $G^\up(P)$ then there is a triangle that contains $p$ and $q$ and no 
other point.  We can shrink such a triangle until $p$ and $q$ are on the boundary 
and at least one of $p$ or $q$ is a corner of the triangle.  
Then $(p,q)$ is an edge of the $\Theta_6$-graph as just defined.
Thus, the above definition of $\Theta_6$-graphs is equivalent to the definition of $G^\both(P)$.
The edges of $G^\up(P)$ come from the odd cones, and the edges of $G^\down(P)$ 
come from the even cones, so the TD-Delaunay graphs $G^\up(P)$ and $G^\down$ 
are known as ``half-$\Theta_6$'' graphs.

TD-Delaunay graphs are called TD-Delaunay ``triangulations''. In fact, they might
fall short of being triangulations.  As discussed by 
Drysdale~\cite{drysdale1990practical} and Chew~\cite{Chew1989} 
(see also~\cite{Aurenhammer14}), they are plane graphs that consist of a 
``support hull'' which need not be convex, and a complete triangulation of the 
interior (an explicit proof can be found in~\cite{Babu2014}).  This anomaly is 
often remedied by surrounding the point set with a large bounding triangle. We 
will use a similar approach later on.    

The $\Theta_6$-graphs, and the more general $\Theta_k$-graphs, which are defined in terms of $k$ cones, have 
some properties that are relevant in a number of application areas.  
In particular, they are \emph{sparse}---$\Theta_k(P)$ has at most $k|P|$ 
edges~\cite{Morin2014}---and they are \emph{spanners}---the ratio (known as 
the \emph{spanning ratio}) of the length of the shortest path between any two 
vertices in $\Theta_k$, $k\geq 4$, to the Euclidean distance between the 
vertices is at most a constant~\cite{Bose2018,Bose2015,Chew1989,Keil1988}.
Because of these properties, $\Theta_k$-graphs have applications in many areas including 
wireless networking~\cite{Alzoubi2003,Bose2012}, motion planning~\cite{Clarkson1987}, 
real-time animation~\cite{Fischer1998}, and 
approximating complete Euclidean graphs~\cite{Chew1989,Keil1992}. 

Among $\Theta_k$-graphs, $\Theta_6$ has some nice properties that make it suitable 
for communications in wireless sensor networks. In particular, $k=6$ is the smallest 
integer for which: 
(i)~$\Theta_k$ has spanning ratio 2~\cite{Bonichon2010,Bose2018,Bose2015}; 
(ii) the so-called $\Theta\Theta_k$-graph, 
which is a subgraph of $\Theta_k$ where each vertex has only one incoming edge per cone, 
is a spanner~\cite{Damian2018}; and 
(iii) so-called half-$\Theta_k$-graphs, which is another subgraph of $\Theta_k$, 
admit a deterministic local competitive routing strategy~\cite{Bose2012}.

\para{Convex Distance Delaunay Graphs.}
For a set $\cal S$ of homothets of a convex polygon, the corresponding proximity 
graphs are the \emph{convex distance Delaunay graphs}.  This concept has been 
thoroughly studied, see, e.g.,~\cite{Aurenhammer14,drysdale1990practical}.  
Some of the helper lemmas we need for half-$\Theta_6$-graphs come from more 
general results that hold for all convex distance Delaunay graphs.

\para{Blocking Sets in Proximity Graphs.} 
Blocking or ``stabbing'' sets were introduced by Aronov \etal\cite{aronov2011witness}
as a more flexible way to represent graphs via proximity
(see also the thesis of Dulieu~\cite{dulieu2012witness}).
The idea was explored further by Aichholzer \etal\cite{Aichholzer2013}
who showed that $3n/2$ points are sufficient and at least $n-1$ points are 
necessary to block any Delaunay triangulation with $n$ vertices. 
Biniaz \etal\cite{Biniaz2015} showed that at least $\left\lceil(n-1)/2\right\rceil$ 
points are necessary to block any $G^\up$-graph with~$n$ vertices.
This bound is tight for $G^\up$-graphs and provides a lower bound on $\beta(n)$. 
The bound also applies to $\Theta_6$-graphs with $n$ vertices, that is, $\beta(n)\geq \left\lceil\frac{n-1}{2}\right\rceil$.
To block any Gabriel graph with $n$ vertices, $n-1$ points are sufficien~\cite{Aronov2013} 
and at least $\left\lceil\frac{n-1}{3}\right\rceil$ points are necessary~\cite{Biniaz2015GG} 
(this lower bound is tight in the sense that there are Gabriel graphs that can be 
blocked by this number of points).

\para{Strong Matchings in Proximity Graphs.} 
The idea of \emph{strong matchings} in proximity graphs---i.e., pairwise disjoint 
objects from $\cal S$ each with two points of~$P$ on the boundary and no points in the interior---was
introduced by \'Abrego \etal\cite{Abrego2004,Abrego2009} for line segments, rectangles, disks, and squares.  
They show that strong (near-)perfect matchings always exist in the first two cases, but that they do 
not always exist for disks (Delaunay graphs) and squares ($L_\infty$-Delaunay graphs). 
In fact, they prove upper bounds of $36n/73 $ and $5n/11$, respectively, on the size of a strong matching.
They also give lower bounds of $\lceil (n-1)/8 \rceil$ and 
$\lceil n/5 \rceil$, respectively.  The lower bound for squares was improved to 
$\lceil (n-1)/4 \rceil$ by Biniaz \etal\cite{Biniaz2015} who also proved lower bounds of
$\lceil (n-1)/9 \rceil$ for $G^\up$ and $\lceil (n-1)/4 \rceil$ for~$G^\both$.

\begin{figure}[t]
	\centering
	\subcaptionbox{\label{cones-fig-a}}{\includegraphics[page=1,scale=.6]{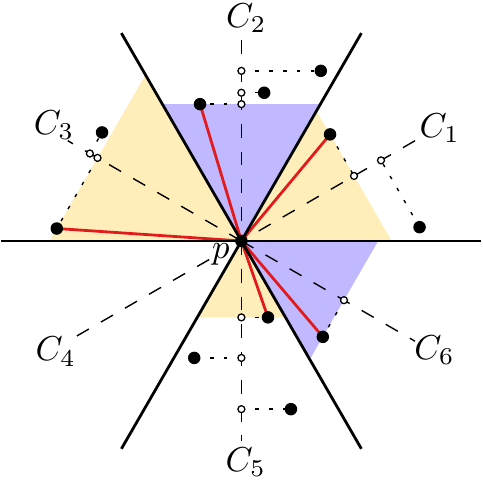}}
	\hfil
	\subcaptionbox{\label{cones-fig-b}}{\includegraphics[page=2,scale=.6]{cones-3}}
	\caption{The construction of (a) the $\Theta_6$-graph, and (b) the odd half-$\Theta_6$-graph.}
	\label{cones-fig}
\end{figure}

\subsection{Preliminaries}\label{subsec:prelim}

We assume that points are in general position and that no line passing through 
two points of $P$ makes an angle of $0^\circ$, $60^\circ$ or $120^\circ$ with the horizontal.

\para{Notation.}
For two points $p$ and $q$ in the plane, we denote by $\bigtriangleup(p,q)$ 
(resp., by $\bigtriangledown(p,q)$) the smallest upward (resp., downward) 
equilateral triangle that has $p$ and $q$ on its boundary. 
We say that a triangle is \emph{empty} if it has no points of $P$ in its interior.
With these definitions, the 
$\Theta_6$-graph has an edge between $p$ and $q$ if and only if $\bigtriangleup(p,q)$ 
is empty or $\bigtriangledown(p,q)$ is empty, in which case we say that the edge $(p,q)$ 
is {\em introduced} by $\bigtriangleup(p,q)$ or by $\bigtriangledown(p,q)$. 

Let $P$ be a set of points.  We use the following notation:
\begin{align*}
    &\mu(P) = \mbox{maximum number of edges in a matching of } G^\both(P)\\
    &\beta(P) = \mbox{minimum size of a set of points that block all empty triangles of } P\\
    &\alpha(P) = \mbox{maximum number of pairwise internally disjoint empty triangles}\\
    &\mu^*(P) = \mbox{maximum number of edges in a strong matching of }  G^\both(P)
\end{align*}
Furthermore, we define $\mu(n), \beta(n), \alpha(n), \mu^*(n)$ to be the minimum of the corresponding parameter over all sets of $n$ points.

\para{Properties of $\Theta_6$-graphs.}  We need the following two properties of $\Theta_6$-graphs:

\begin{lemma}[Babu \etal\cite{Babu2014}]
	\label{path-in-triangle-lemma}
	Let $P$ be a set of points in the plane, and let $p$ and $q$ be any two points in $P$. There is a path between $p$ and $q$ in $G^\up(P)$ that lies entirely in $\bigtriangleup(p,q)$. Moreover, the triangles that introduce the edges of this path also lie entirely in $\bigtriangleup(p,q)$. Analogous statements hold for $G^\down(P)$ and $\bigtriangledown(p,q)$.  
\end{lemma}

We remark that this lemma holds more generally for any convex-distance Delaunay graph.  
The second property we need has been proved in the general setting of convex-distance Delaunay graphs.  
It generalizes the fact that the (standard) Delaunay triangulation contains the 
minimum spanning tree with respect to Euclidean distances.  We state the result 
for the special case of equilateral triangles. 
For any two points $p$ and $q$ in the plane, define the weight function $w^\up(p,q)$ 
to be the area of the smallest $\up$-triangle containing $p$ and $q$.

\begin{lemma}[Aurenhammer and Paulini~\cite{Aurenhammer14}]
\label{intersection-lemma}
   The minimum spanning tree of points~$P$ with respect to the weight function $w^\up(p,q)$ is contained in $G^\up(P)$. 
\end{lemma}

A consequence of Lemma~\ref{intersection-lemma} (as noted by Aurenhammer and 
Paulini in their more general setting) is that 
the minimum spanning tree of points $P$ with respect to the weight function $w^\up(p,q)$ 
is contained in both $G^\up(P)$ and $G^\down(P)$, because $w^\up(p,q) = w^\down(p,q)$.
In particular, this means that the intersection of $G^\up(P)$ and $G^\down(P)$ 
is connected, as was proved with a different method by Babu \etal\cite{Babu2014}.

\para{The Tutte-Berge Matching Theorem.}
Let $G$ be a graph and let $S$ be an arbitrary subset of vertices of $G$. 
Removing $S$ splits $G$ into a number, $\comp(G \setminus S)$, of connected components.
Let $\odd(G\setminus S)$ denote the number of odd components of $G\setminus S$, i.e., 
components with an odd number of vertices. In 1947, Tutte~\cite{Tutte1947} 
characterized graphs that have a (near-)perfect matching as exactly those graphs 
that have at most $|S|$ odd components for any subset $S$. In 1957, Berge~\cite{Berge1958} 
extended this result to a formula (today known as the Tutte-Berge formula) for 
the size of maximum matchings in graphs. The following is an alternate way of 
stating this formula in terms of the number of unmatched vertices, i.e., vertices 
that are not matched by the matching.  

\begin{theorem}[Tutte-Berge formula; Berge~\cite{Berge1958}]
	\label{tutte-berge-thr}
	The number of unmatched vertices of a maximum matching in $G$ is equal to the maximum over subsets $S \subseteq V$ of $\odd(G \setminus S)-|S|$.  
\end{theorem}

To obtain a lower bound on the size of a maximum matching it suffices, by 
Theorem~\ref{tutte-berge-thr}, to find an upper bound on $\odd(G \setminus S)-|S|$ 
that holds for any $S$. We will use this approach in our proofs of 
Theorems~\ref{thm:matching-bound} and~\ref{thm:reductions}.
In fact, as in Dillencourt's proof~\cite{Dillencourt1990} that Delaunay graphs 
have perfect matchings we will find an upper bound on $\comp(G \setminus S)-|S|$ 
that holds for any $S$, i.e., we establish a bound on the \emph{toughness} of 
the graph~\cite{bauer2006toughness}.

\section{Bounding the Size of a Matching}\label{sec:matching-bound}

In this section, we prove Theorem~\ref{thm:matching-bound}.  
Let $P$ be a set of $n$ points in the plane and let $G^\both(P)$ be the $\Theta_6$-graph on $P$.
We will prove that $G^\both(P)$ contains a matching of size at least $(3n-8)/7$. 
As implied by Theorem~\ref{tutte-berge-thr}, in order to prove a lower bound on the size of maximum matching in $G^\both(P)$, 
it suffices to prove an upper bound on $\odd(G^\both(P) \setminus S)-|S|$ that holds for any subset $S$ of $P$. 
Since it is hard to argue about odd components, we will 
in fact prove an upper bound on $\comp(G^\both(P) \setminus S) - |S|$.  Such a 
bound applies to $\odd(G^\both(P) \setminus S)-|S|$ 
because $\odd(G^\both(P) \setminus S)\leq \comp(G^\both(P) \setminus S)$.  

Our proof will depend on an analysis of the faces of $G^\up(P) \setminus S$ and
$G^\down(P) \setminus S$ for which we need some preliminary results.
Consider a planar graph $G$ with a fixed planar embedding.
Such an embedding divides the plane
into connected regions, called {\em faces}. For every face~$f$ of~$G$, we define its {\em degree} 
as the number of triangles in a triangulation of $f$ plus 2; see Figure~\ref{deg-fig} for some examples. 
A similar notion of degree has been used in~\cite{Biedl2004}. We emphasis that 
we do not really add any edges to $G$; these edges are imaginary, just to define the degree of a face.  
Let ${\cal F}_d(G)$ denote the set of faces of $G$ of degree $d$. 
An easy counting argument shows that if $|V|\geq 3$, then 
$\sum_{d\geq 3} (d-2) |{\cal F}_d(G)|=  2|V|-4$,
since a face of degree~$d$ gives rise to $d-2$ faces in a triangulation of $G$, which has $2|V|-4$ faces.

\begin{figure}[t]
	\centering
	\includegraphics[width=.8\columnwidth]{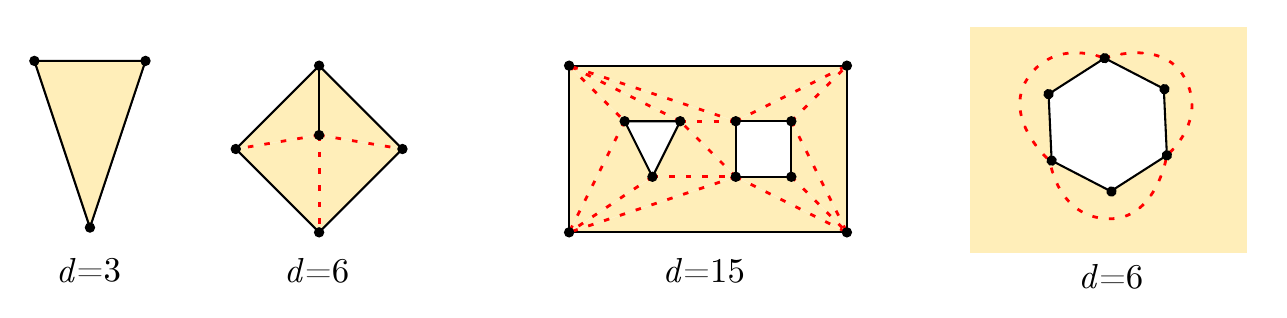}
	\caption{The notion of degree of a face.  }
	\label{deg-fig}
\end{figure}

We will utilize the following lemma that Dillencourt used in his proof that every Delaunay triangulation contains a (near-)perfect matching.
Let $G[S]$ denote the subgraph of $G$ that is induced by a subset $S$ of its vertices.

\begin{lemma}[Dillencourt~\cite{Dillencourt1990}, Lemma 3.4]
\label{lem:face_double}
Let $G$ be a triangulated planar graph and let $S$ be a subset of 
vertices of $G$. Then every face of $G[S]$ contains at most one component of~$G\setminus S$.
\end{lemma}

We aim to apply this result to $G^\up(P)$ and $G^\down(P)$. 
As noted in Section~\ref{subsec:background}, the interior faces of $G^\up(P)$ 
and $G^\down(P)$ are triangles, but their outer faces need not be the convex hull of~$P$. 
For this reason, and also for Lemma~\ref{lem:f3_4} below,
we add a set $A=\{a_1,\dots,a_6\}$ of 
{\em surrounding points} as follows.  Find the smallest 
$\up$-triangle~$T^\up$ and $\down$-triangle~$T^\down$ containing
all points of~$P$.  Let $\mathcal{R}(P)$ be the region $T^\up\cup T^\down$.
(we will need this definition again in Section~\ref{sec:blocking}). 
Observe that all of the empty triangles that introduce edges
of $G^\both(P)$ lie in $\mathcal{R}(P)$, so adding points outside $\mathcal{R}(P)$ does not 
remove any edge from the graph.  
We now place points $a_1,\dots,a_6$ near the corners of $T^\up$  and $T^\down$ 
(see Figure~\ref{augment-fig}): at each corner, place a point in the cone opposite to the cone that contains the triangle,
and name the points in such a way that every point of $P$ has $a_i$ in cone $C_i$.  

Now fix a set $S$ for which we want to bound $\comp(G^\both(P)\setminus S) - |S|$, and define $S_A = S\cup A$.
Pick an arbitrary representative point from every connected component of 
$G^\both(P)\setminus S$, and let $\RC$ be the set of these points,
so $|\RC|=\comp(G^\both(P)\setminus S)$.

Define $G_A^\up = G^\up(P\cup A)$  and consider its subgraph
$G_A^\up[S_A]$ induced by $S_A$.  By construction, the outer face of 
both $G_A^\up$ and $G_A^\up[S_A]$ is the hexagon formed by $A$;
we add three graph edges (not segments) to triangulate the outer face,   
so that $G_A^\up$ is triangulated.  
Note that none of the points of $P$ (and in particular therefore no points of $\RC$) are inside
the four newly introduced triangular faces.

\begin{figure}[t]
  \centering
  \subcaptionbox{\label{augment-fig}}{\includegraphics[page=1,scale=.75]{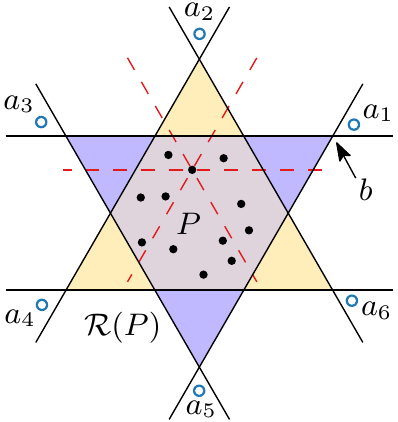}}
  \hfil
  \subcaptionbox{\label{lemf3_4-fig}}{\includegraphics[page=1,scale=.75]{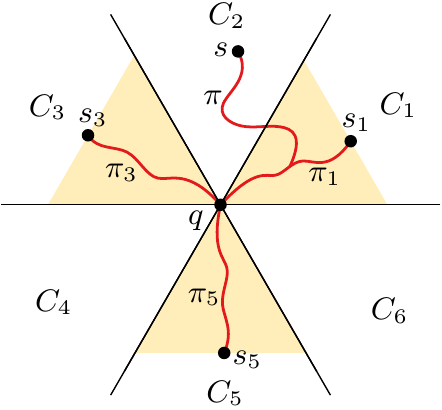}}
	\caption{(a) Augmentation of $P$: the shaded region is $\mathcal{R}(P)$, and $A=\{a_1,\dots,a_6\}$. (b)~Illustration for the proof of Lemma~\ref{lem:f3_4}.}
\end{figure}

Let $f_d^\up$ be the number of faces of degree $d$ in $G_A^\up[S_A]$ that contain some point of~$\RC$. 
We define $f_{4+}^\up=\sum_{d\ge 4} f_d^\up$. 
Since all faces of $G_A^\up$ are now triangles, (after we added those edges),
Lemma~\ref{lem:face_double} applies and
every face of $G_A^\up[S_A]$ contains at most one component, hence at most one point of $\RC$. Therefore,
\begin{equation}
\label{eq2}
|\RC|=f_3^\up + f_{4+}^\up\text{\quad\quad and similarly\quad\quad}|\RC|=f_3^\down + f_{4+}^\down,
\end{equation}
where $f_d^\down$ is defined in a symmetric manner on graph $G_A^\down[S_A]$.

Let ${\cal F}_d$ be the set of faces of degree $d$ in $G_A^\up[S_A]$ and observe 
that, since no point of $\RC$ appears in the four triangles outside
the hexagon of $A$, we have $f_3^\up \leq |{\cal F}_3| - 4$.  As a consequence,
\begin{align}
\label{eq3}
f_3^\up + 2 f_{4+}^\up &\leq 
\sum_{d\geq 3} (d-2)f_d^\up \leq 
\sum_{d\geq 3} (d-2)|{\cal F}_d| - 4\nonumber\\
&\leq 2|V(G_A^\up[S_A])|-4 -4 = 2|S|+2|A|-8 = 2|S|+4
\end{align}
and similarly
$f_3^\down + 2 f_{4+}^\down \leq 2|S|+4$.

The crucial insight for getting an improved matching bound is that no component 
can reside inside a face of degree 3 in both $G^\up$ and $G^\down$.
Formally, we show:  
   
\begin{lemma}
\label{lem:f3_4}
We have $f_3^\triangle\leq f_{4+}^\down$ and $f_3^\down \leq
f_{4+}^\triangle$.
\end{lemma}

\begin{proof}
Consider any point $\rc\in \RC$, hence $\rc\not\in S_A$. 
Let $F^\up$ and $F^\down$ be the faces of~$G_A^\up[S_A]$ and $G_A^\down[S_A]$ that contain $\rc$, respectively. 
It suffices to show that one of~$F^\up$ and~$F^\down$ has degree at least 4.

By Lemma~\ref{intersection-lemma}, the minimum-weight spanning tree $T$ of $P\cup A$ 
belongs to both~$G_A^\up$ and~$G_A^\down$. Find a path $\pi$ 
in $T$ that connects $\rc$ to some point $s\in S_A$
such that no vertex of $\pi$ except $s$ belongs to $S_A$.

Assume first that $s$ is in a cone with even index.
Let $s_1, s_3,s_5$ be the points of $S_A$ that are closest to $\rc$ in 
cones $C_1,C_3,C_5$, respectively; since $A\subseteq S_A$, such points $s_i$ exist. 
Refer to Figure~\ref{lemf3_4-fig}. By Lemma~\ref{path-in-triangle-lemma}, for 
every $i\in\{1,3,5\}$, there exists a path $\pi_i$ between $\rc$ and $s_i$ 
in~$G^\up$ that lies fully in $\bigtriangleup(\rc,s_i)$. By our choices of $s_i$, 
no vertex of $\pi_i$ except $s_i$ is in $S_A$.
  
So we have four (not necessarily disjoint) paths $\pi,\pi_1,\pi_3,\pi_5$ in $G^\up$ 
that begin at $\rc$ and end at four points $s,s_1,s_3,s_5$ of $S_A$. 
These points are distinct because they belong to four different cones of $\rc$. 
Furthermore, intermediate points of these paths are not in $S_A$. 
This implies that $s, s_1, s_3, s_5$ belong to the boundary of the same face $F^\up$ of $G^\up[S_A]$.
In consequence, $F^\up$ has degree at least 4.

Similarly, if $s$ is in a cone with odd index, 
then $F^\down$ has degree at least 4, proving the claim.
\end{proof}

Now we have tools to prove an upper bound on the number of unmatched vertices and, 
more generally, the toughness of a $\Theta_6$-graph. 

\begin{lemma}
For any $S\subseteq P$, we have $\comp(G^\both(P)\setminus S) - |S| \leq  (|P|+16)/7$.
\end{lemma}
\begin{proofWithFormula}
Recall that we fixed a set $\RC$ of points in $P\setminus S$ with $|\RC|=\comp(G^\both(P)\setminus S)$.
So $n=|P|\geq |S|+|\RC|$, or equivalently $n-|\RC|-|S|\geq 0$.
Combining this with the above inequalities, we get
\begin{align*}
	7\left(\comp(G^\both(P)\setminus S)-|S|\right) 
	& \leq  7|\RC| - 7|S| + (n-|\RC|-|S|) &\\
	& =  n + 3|\RC| + 3|\RC| - 8|S| &\\
	\mbox{(by \eqref{eq2})\quad\qquad} & =  n + 3\left(f_3^\up+f_{4+}^\up\right) + 3\left(f_3^\down+f_{4+}^\down\right)  - 8|S| \\
	\mbox{(by Lemma~\ref{lem:f3_4})\quad\qquad} & \leq  n + 2f_3^\triangle+4f_{4+}^\triangle + 2f_3^\down+4f_{4+}^\down  - 8|S| \\
	\mbox{(by \eqref{eq3})\quad\qquad} & \leq  n + (4|S|+8) + (4|S|+8) - 8|S| \\
		& =  n + 16. \tag*{\qed}
\end{align*}
\end{proofWithFormula}

Therefore, $\odd(G^\both(P)\setminus S)-|S| \leq \comp(G^\both(P)\setminus S)-|S| \leq(n+16)/7$. 
In consequence of the Tutte-Berge formula, 
therefore any maximum matching $M$ of $G^\both(P)$ has at most $(n+16)/7$ unmatched
vertices, hence at least $(6n-16)/7$ matched vertices and $|M|\geq (3n-8)/7$.
This completes the proof of Theorem~\ref{thm:matching-bound}.

\begin{figure}[t]
	\centering
	\includegraphics[scale=.75]{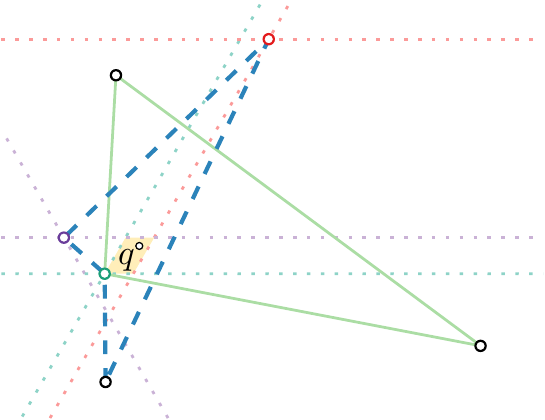}
	\caption{A set $P$ of seven points and a subset $S$ (the six larger points). 
    The graph $G^\both(P)\setminus S$ contains a singleton-component $\rc$ which 
    lies in a face of degree 3 in $G^\down[S]$ (green solid edges) and a face of 
    degree 4 in $G^\up[S]$ (blue dashed edges).}
	\label{3-4-fig}
\end{figure}

\para{Remark.}
If we knew $f_3^\triangle\leq f_{5+}^\down$ and $f_3^\down \leq f_{5+}^\triangle$ 
(where $f_{5+}^\up=\sum_{d\geq 5} f_d^\up$ etc.), then a similar analysis would 
show $\odd(G^\both(P)\setminus S)-|S|\leq 4$, which would imply 
Conjecture~\ref{conj:perfect-matching} except for a small constant term. 
However, Figure~\ref{3-4-fig} shows an example where a point $\rc\not\in S$ lies 
in a face of degree 3 in $G_A^\down$ and a face of degree 4 in $G_A^\up$, so our 
proof-approach cannot be used to prove such a claim.

\section{The Relationship Between Blocking Sets and Matchings}\label{sec:blocking}

In this section, we prove Theorem~\ref{thm:reductions}---that a lower bound on 
the blocking size function $\beta(n)$ implies a lower bound on the 
size $\mu(n)$ of a maximum matching, and vice versa.

\begin{lemma}
	\label{increasing-lemma}
	For any $n\geq 1$, we have $\beta(n+1)\leq \beta(n)+1$.
\end{lemma}
\begin{proof}
Consider a set $P$ with $n$ points such that $\beta(n)=\beta(P)$. 
Let $T^\down$ be a downward equilateral triangle that strictly encloses all points of $P$.
Let $b$ be the rightmost point of~$T^\down$.  Then $P$ lies in cone $C_4$ of $b$.
Let $a_1$ be a point strictly inside cone $C_1$ of $b$;
see also Figure~\ref{augment-fig}.
Every upward or downward equilateral triangle between $a_1$ and any point of $P$ contains 
the point $b$.
Set $P'=P\cup \{a_1\}$, and observe that we can block $G^\both(P')$ by using a minimum blocking set~$B$ of $G^\both(P)$ and adding $b$ to it.  
Since $|B|=\beta(P)=\beta(n)$, we have $\beta(P')\leq \beta(n)+1$, and $\beta(n+1)$ cannot be larger than that.
\end{proof}

Since $\beta(1)=0$, this lemma also shows that $\beta(n)\leq n-1$, or in other words, that the `$n-1$' in Conjecture~\ref{conj:blocking} is tight.
We are now ready to prove Theorem~\ref{thm:reductions}~(a). 

\begin{backInTime}{thmReductions}
  \begin{theorem}[a]
    \label{thm:blocking-to-matching}
    For any set $P$ of $n$ points in the plane, $G^\both(P)$ has a matching of size $\beta(n)/2$.
  \end{theorem}
\end{backInTime}
\begin{proofWithFormula}
	Consider the $\Theta_6$-graph $G^\both(P)$ on a set $P$ of $n$ points in the plane. 
	We again use the Tutte-Berge formula (Theorem~\ref{tutte-berge-thr}) to prove that $G^\both(P)$ contains a matching of size at least $\beta(n)/2$. 
	Fix an arbitrary set $S \subseteq P$ and consider the connected components of $G^\both(P)\setminus S$.  As in the proof of Theorem~\ref{thm:matching-bound},  fix one representative point in each component, and  let
	$\RC$ be the set of these points.  
	
  Consider the $\Theta_6$-graph $G^\both(\RC)$ of only the points in $\RC$, and 
  let $(\rc_1,\rc_2)$ be an edge in it; say it is introduced by $\bigtriangleup(\rc_1,\rc_2)$.  	
	By Lemma~\ref{path-in-triangle-lemma}, there is a path $\pi$ between $\rc_1$ 
  and $\rc_2$ in~$G^\up(P)$ that is fully contained in $\bigtriangleup(\rc_1,\rc_2)$; 
  moreover, all triangles introducing the edges of $\pi$ lie in $\bigtriangleup(c_1,c_2)$. 
	Since $\rc_1$ and $\rc_2$ are in different components of $G^\both(P)\setminus S$,
  at least one point of $\pi$ belongs to~$S$.   
	
	Thus, for any edge in $G^\both(\RC)$, the triangle that supports that edge 
  contains a point in~$S$.  Put differently,$S$ blocks $G^\both(\RC)$, and 
  thus $|S|\geq \beta(|\RC|)$. Furthermore, $\beta(n)\leq \beta(|\RC|)+n-|\RC|$ 
  by Lemma~\ref{increasing-lemma} since $|\RC|\le n$. Combining this with Theorem~\ref{tutte-berge-thr}, 
	it follows that the size of maximum matching in $G^\both(P)$ is at least
	\[\frac{n-(|\RC|-|S|)}{2}\geq\frac{n-(|\RC|-\beta(|\RC|))}{2}\geq\frac{n-(n-\beta(n))}{2}=\frac{\beta(n)}{2}. \tag*{\qed}\]  
\end{proofWithFormula}

In particular, if $\beta(n)\geq n-1$, then $\mu(n)\geq \beta(n)/2 \geq (n-1)/2$, 
so by integrality $\mu(n) \geq \lceil (n-1)/2 \rceil$.  In other words, 
Conjecture~\ref{conj:blocking} implies Conjecture~\ref{conj:perfect-matching}.

We now turn to the other half of Theorem~\ref{thm:reductions}.
Note that Aichholzer \etal\cite{Aichholzer2013} proved a similar result 
(for $c=d=1/2$ and Delaunay graphs), and our proof is a modification of theirs.  
(In fact, the proof applies to any proximity graphs.)

\begin{backInTime}{thmReductions}
  \begin{theorem}[b]
    \label{thm:matching-to-blocking}
    Assume that we know that $\mu(n) \geq cn+d$ for some constants $c,d$.
    Then $\beta(n)\geq (cn+d)/(1-c)$.
  \end{theorem}
\end{backInTime}

\begin{proof}
Let $P$ be a set of $n$ points such that $\beta(P)=\beta(n)=b$, and let $B$ be 
a minimum blocking set of $G^\both(P)$.
Then $P$ is an independent set in $G^\both(P \cup B)$.
Let $M$ be a matching of size at least $\mu(b+n)\geq cb+cn+d$ in $G^\both(P\cup B)$.  
Since $P$ is an independent set in $G^\both(P \cup B)$, it contains at most one 
endpoint of each edge in $M$, as well as some unmatched points, so 
\[ n = |P| \leq |M| + (n+b-2|M|) \leq  n + b - (cb+cn+d)   \] 
Solving for $b$ gives $\beta(n) = b \geq (cn+d)/(1-c)$.
\end{proof}

In particular, if Conjecture~\ref{conj:perfect-matching} holds, then $\mu(n)\geq (n-1)/2$.
Hence, $c=1/2$ and $d=-1/2$, therefore $\beta(n) \geq 2(n-1)/2 = n-1$ and Conjecture~\ref{conj:blocking} holds. 
So Conjecture~\ref{conj:perfect-matching} implies Conjecture~\ref{conj:blocking}. 
As a second consequence, we know that $(3n-8)/7$ is a valid lower bound on $\mu(n)$ 
by Theorem~\ref{thm:matching-bound}, therefore (with $c=3/7$) we have 
$\beta(n)\geq 7/4\cdot(3n-8)/7 = 3n/4-2$.

\section{\texorpdfstring{Other Bounds on $\alpha$, $\mu^*$, and $\beta$.}{Other Bounds}}
\label{sec:disjoint-triangles}

In this section, we give upper bounds on $\alpha(n)$ and $\mu^*(n)$.  
Specifically, we give an example of $n$ points for which the maximum number of 
pairwise internally disjoint empty triangles is $3n/4$; this shows that $\alpha(n) \le 3n/4$.
Then we give an example on $n$ points for which the maximum strong matching 
has $2n/5$ edges; this shows that $\mu^*(n) \le 2n/5$.

We defined $\beta(n)$ to be the minimum size of a blocking set of any 
$\Theta_6$-graph on $n$ points because this was relevant for matchings, but it 
is also interesting to know the maximum number of points that may be needed to 
block any $\Theta_6$-graph on $n$ points, i.e., to establish bounds 
on $\hat\beta(n)$---the maximum, over all points sets~$P$ of size $n$, 
of $\beta(P)$. An easy upper bound on $\hat\beta(n)$ follows from   
Biniaz \etal\cite{Biniaz2015} who showed that $G^\down$ can always be blocked 
by $n-1$ points placed just above every input point except for the topmost one. 
By symmetry, $G^\up$ can always be blocked by $n-1$ points, and thus, 
$G^\both$ can be blocked by at most $2(n-1)$ points, i.e., $\hat\beta(n) \le 2(n-1)$. 
Our final example of this section is a set of points $P$ such 
that $\alpha(P) \ge (5n-6)/4$, and thus $\beta(P) \ge (5n-6)/4$; this shows 
that $\hat\beta(n) \ge (5n-6)/4$.

\para{An upper bound on $\alpha(n)$} 

Figure~\ref{disjoint-4cluster} shows how to construct a point set of size $n$ such that $\alpha(P) = 3n/4$.
The point set consists of repeated copies of a cluster $R$ of four points arranged as shown in 
Figure~\ref{disjoint-4cluster-a}. 
Observe that there are 8 empty triangles formed by pairs of points in~$R$: 2~for 
each of the three dashed edges, and 1 for each of the two long blue edges---we 
call these the ``blue triangles''.  

\begin{lemma}
  In $R$, there are at most 3 interior-disjoint empty triangles.
  \label{lem:3-disjoint}
\end{lemma}
\begin{proof}
  If neither blue triangle is used, then there are at most 3 interior-disjoint 
  triangles, one for each dashed edge.  
  Using both blue triangles rules out all other empty triangles. Using exactly 
  one blue triangle rules out both empty triangles corrsponding to the black dotted edge.
\end{proof}

The final configuration consists of~$t$ copies $R_1, \ldots, R_t$ of $R$, called 
\emph{clusters}, where $R_{i+1}$ lies in cone~$C_1$ of all the points of $R_i$. 
If we do not use empty triangles determined by pairs of points from different clusters,
then by Lemma~\ref{lem:3-disjoint} we can get at most 3 empty triangles for each 4 
points in $R_i$ for a total of $3n/4$ interior disjoint empty triangles.  
It remains to analyze what happens when we use empty triangles between different clusters. 

Consider an empty triangle $T$ determined by two points $p$ and $q$ in different clusters. 
Then the points lie in consecutive clusters, say $R_{i-1}$ and $R_{i}$.  Furthermore, 
one of the points, say $p$, lies at a corner of $T$.  We assign the triangle $T$ to 
the cluster of the other point $q$.  Observe (see Figure~\ref{disjoint-4cluster-b})
that $q$ must be the unique extreme point of its cluster, but point $p$ is not unique.
The proof that the point set allows at most $3n/4$ interior-disjoint empty triangles follows from the following lemma.

\begin{figure}[t]
	\centering
	\subcaptionbox{\label{disjoint-4cluster-a}}{\includegraphics[scale=.2]{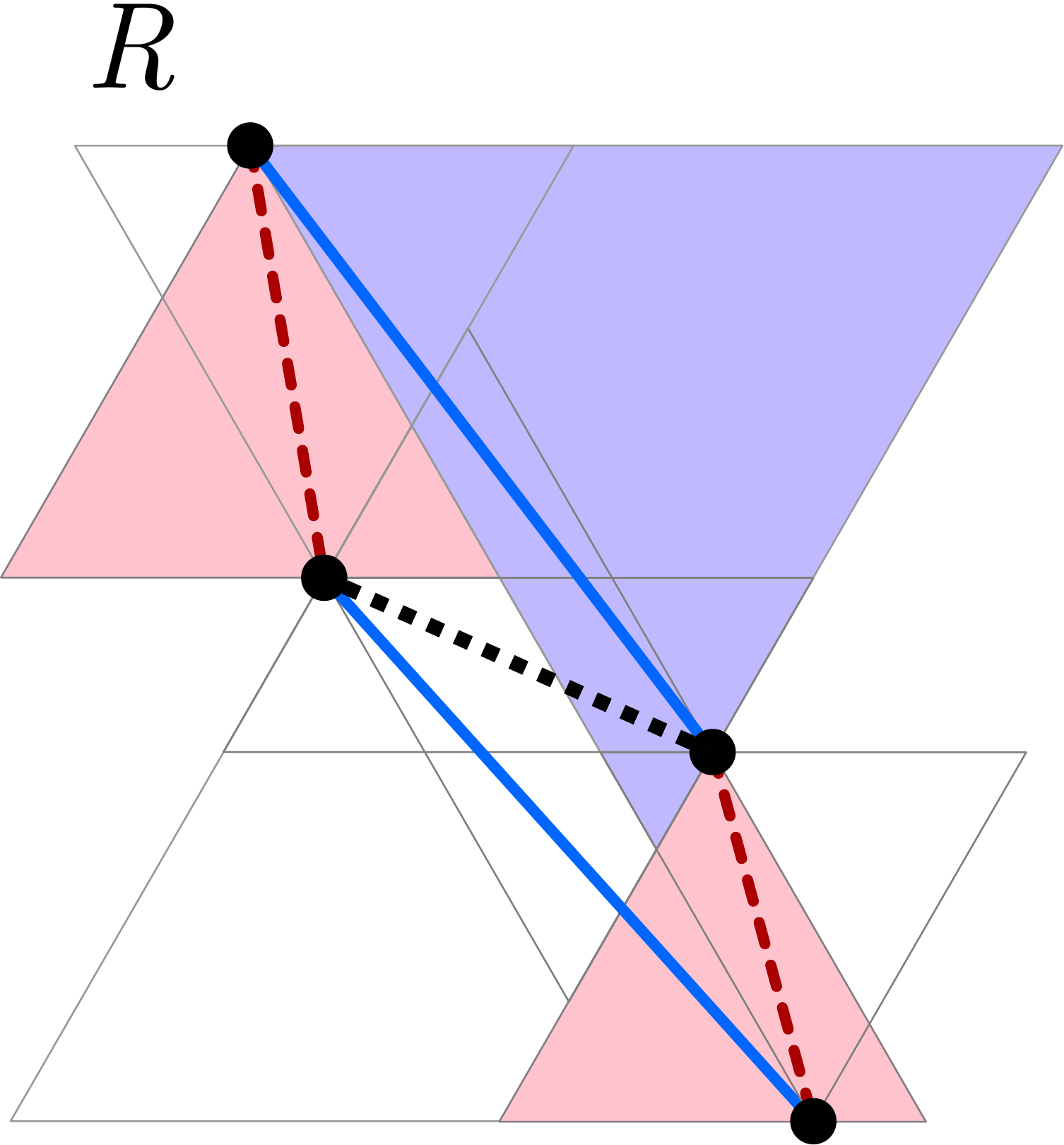}}
  \hfil
	\subcaptionbox{\label{disjoint-4cluster-b}}{\includegraphics[scale=.15]{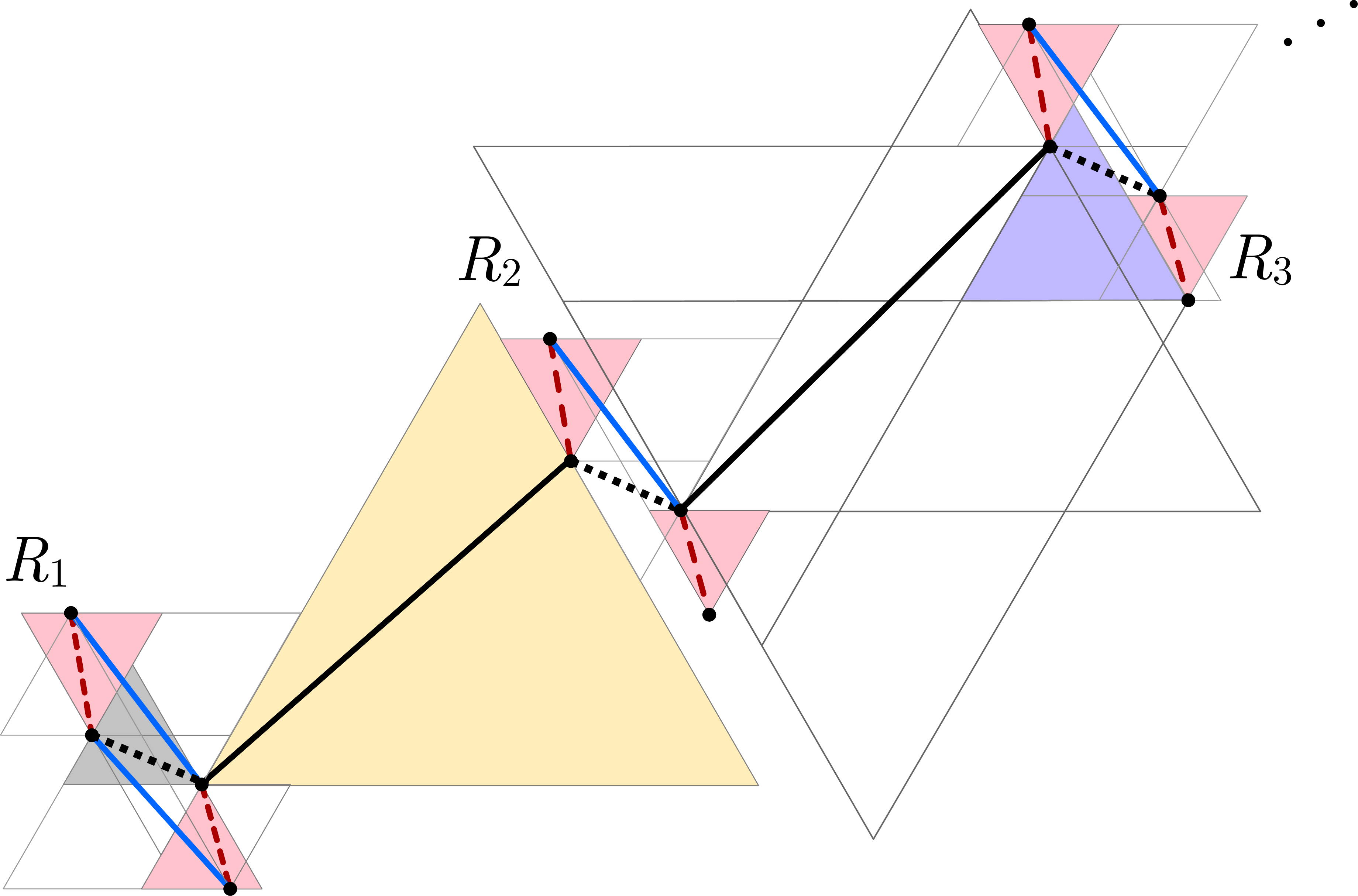}}
	\caption{A point set for which the maximum number of disjoint triangles is $3n/4$. 
    (a) The basic cluster $R$, its $\Theta_6$-graph, and a set of 3 possible 
    internally-disjoint empty triangles. 
    (b)~The final point set formed by repeating $R$.  Only some of the empty 
    triangles between clusters are shown.}
	\label{disjoint-4cluster}
\end{figure}

\begin{lemma}
    For any set of interior-disjoint empty triangles and any $i$, there are at 
    most 3 triangles assigned to or contained in $R_i$.
\end{lemma}
\begin{proof}
    Consider $R_i$ and suppose that our set contains one between-cluster empty 
    triangle assigned to $R_i$.  By symmetry, we may suppose that this triangle 
    has a corner at a point  in~$R_{i-1}$; see, for example, the large yellow 
    triangle in Figure~\ref{disjoint-4cluster-b}. This triangle intersects~4 of 
    the empty triangles of $R_i$, and it is easy to check that there are at
    most 2 internally-disjoint triangles left. 
    
    Next, suppose that we use more than one between-cluster empty triangle 
    assigned to $R_i$.  Then there must be exactly two such triangles, one with 
    a corner in $R_{i-1}$ and one with a corner in $R_{i+1}$.  But then all the 
    empty triangles inside~$R_i$ are ruled out.
\end{proof}

\para{An upper bound on $\mu^*(n)$.}

Figure~\ref{strong-matching-fig} shows how to construct a point set~$P$ of 
size $n$ such that $\mu^*(P) = 2n/5$. The point set consists of repeated copies 
of a cluster $S$ of five points arranged as shown in Figure~\ref{strong-matching-fig-a}. 
It is crucial that the two triangles shown in the figure intersect.
The final configuration consists of $t$ copies $S_1,\ldots,S_t$ of $S$, again 
called clusters, where $S_{i+1}$ lies in cone $C_1$ of all the points of $S_i$. 
See Figure~\ref{strong-matching-fig-c}.

\begin{figure}[t]
	\centering
	\subcaptionbox{\label{strong-matching-fig-a}}{\includegraphics[scale=.25]{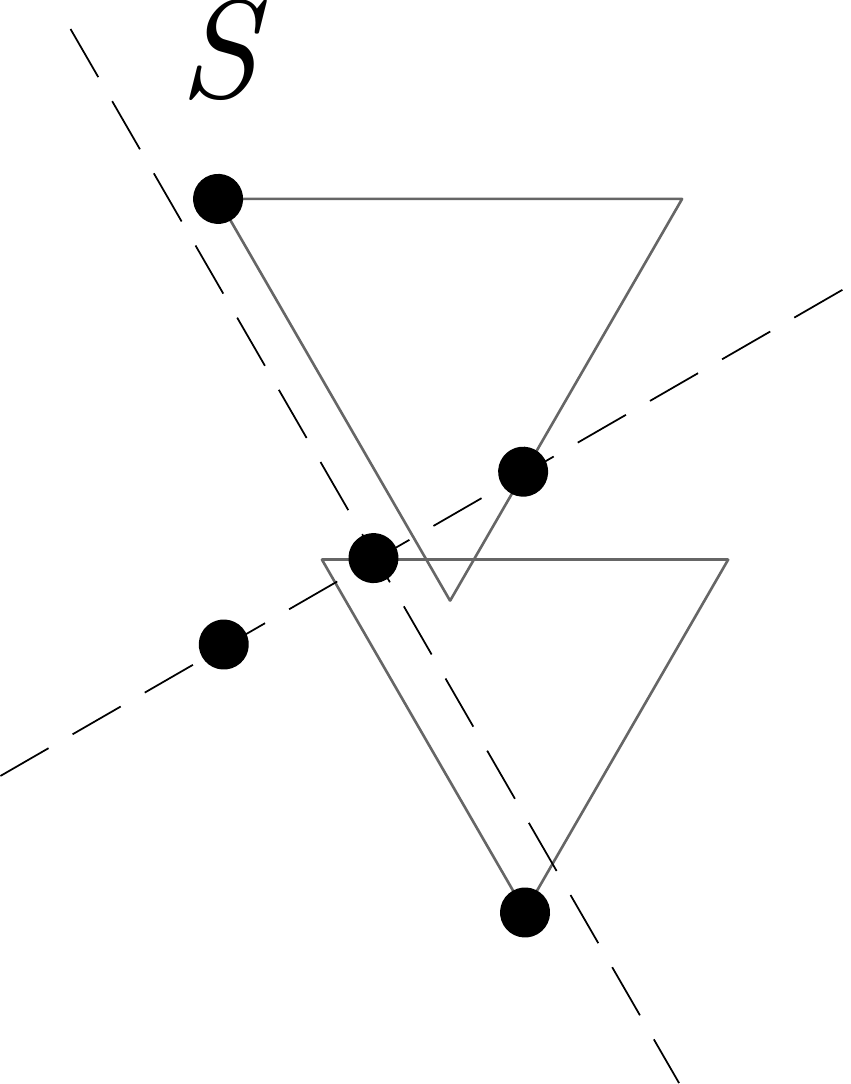}}
  \hfill
  \subcaptionbox{\label{strong-matching-fig-b}}{\includegraphics[scale=.25]{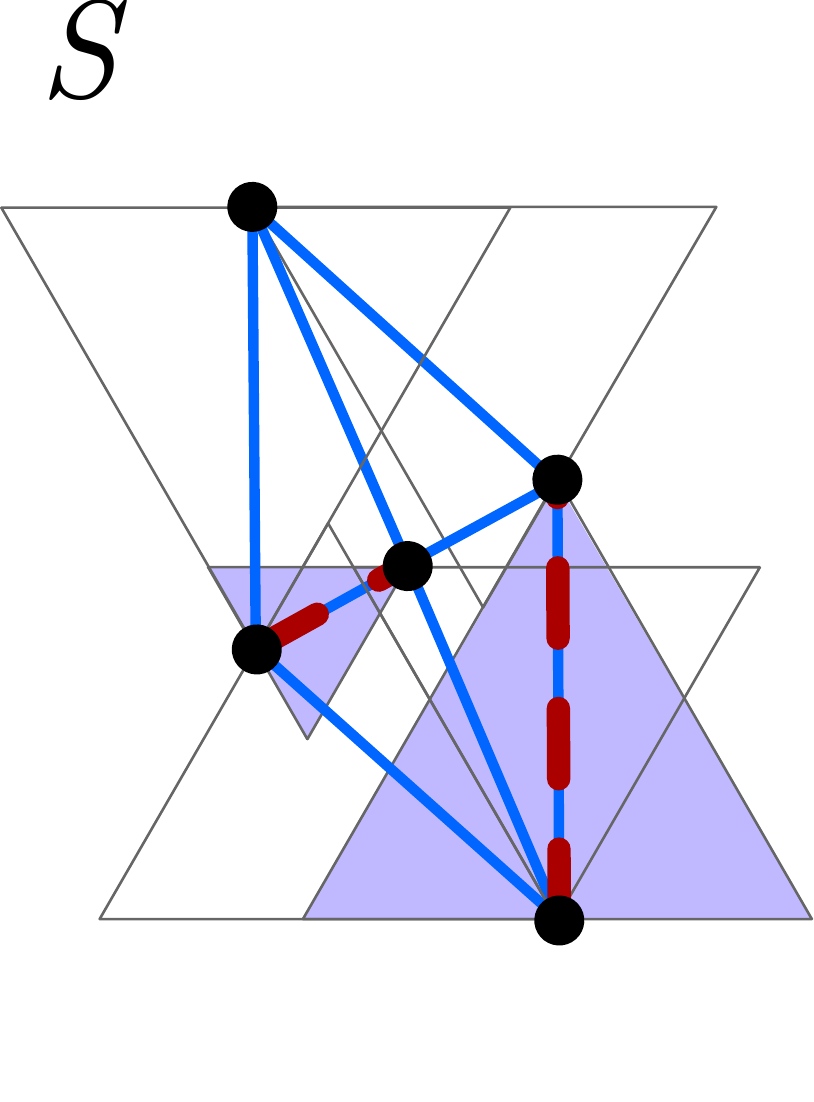}}
  \hfill
  \subcaptionbox{\label{strong-matching-fig-c}}{\includegraphics[scale=.2]{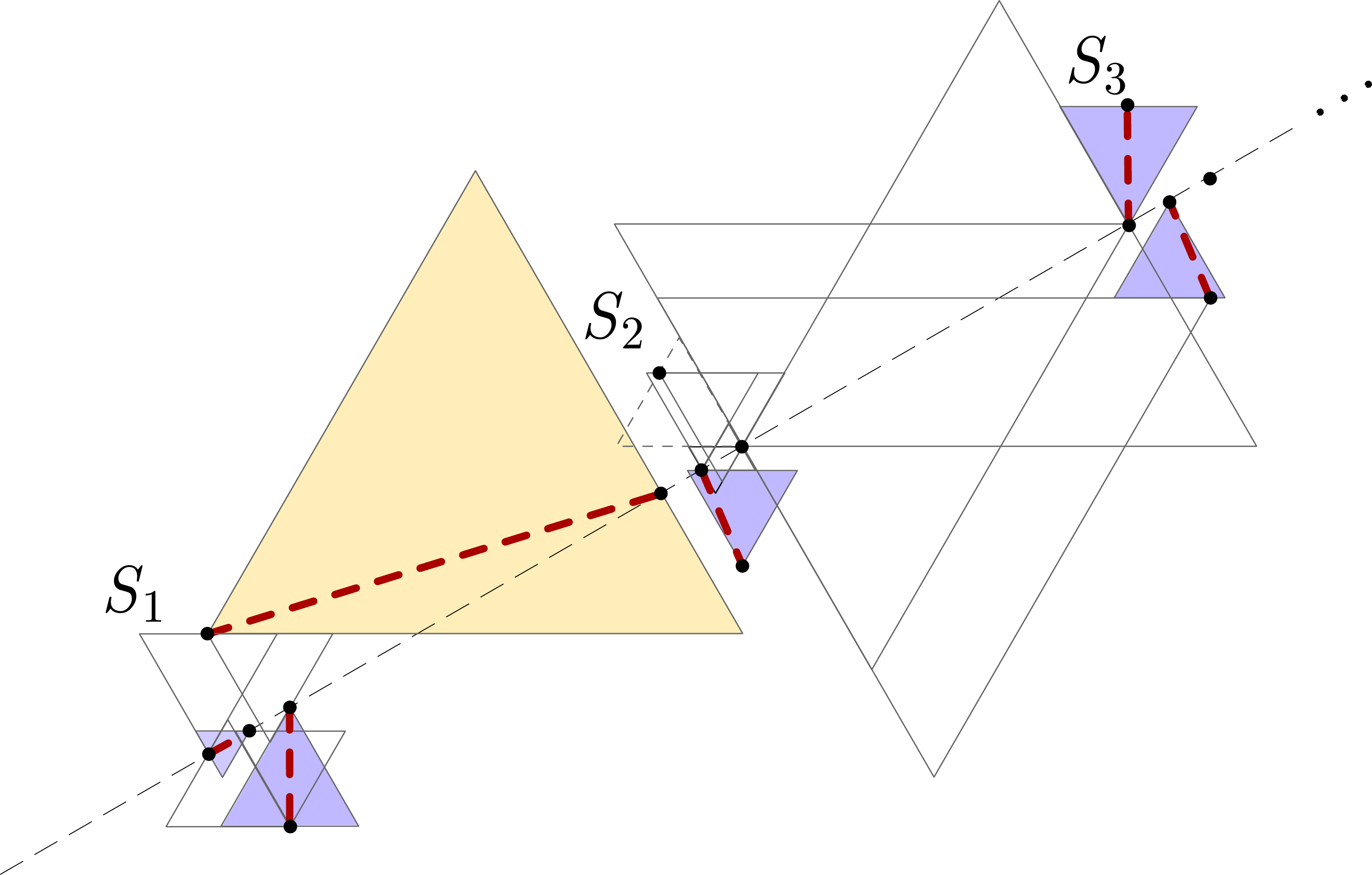}}
	\caption{A point set for which the maximum strong matching has at most $2n/5$ edges. 
    (a)~The basic cluster $S$ of 5 points. The dashed lines are guide lines 
    at $30^\circ$ and $120^\circ$.  
    (b) The $\Theta_6$-graph of~$S$, some of the empty triangles, and a strong matching (dashed red). 
    (c) The final point set formed by repeating $S$.  Only some of the empty 
    triangles between clusters are shown.}
	\label{strong-matching-fig}
\end{figure}

If we do not use empty triangles between clusters, then each cluster has at most 
two disjoint empty triangles, i.e., at most two strong matching edges, so the 
matching has at most $2n/5$ edges. As in the previous construction, an empty triangle $T$
determined by two points $p$ and $q$ in different clusters must go between 
consecutive clusters, and one point, say~$p$, must lie at a corner of $T$. 
As before, we assign such a triangle to the cluster containing the other point $q$.  
The proof that the point set allows at most $2n/5$ strong matching edges follows from the following lemma.

\begin{lemma}
    For any set of disjoint empty triangles and any $i$, there are at most 2 triangles assigned to or contained in $S_i$.  
\end{lemma}
\begin{proof}
  Consider $S_i$ and suppose that our set contains one between-cluster empty
  triangle assigned to $S_i$.  By symmetry, we may suppose that this triangle has 
  a corner at a point  in $S_{i-1}$; see, for example, the large yellow triangle 
  in Figure~\ref{strong-matching-fig-c}. There are only 4 points and~5 empty 
  triangles in $S_i$ that are disjoint from the big triangle (these are shown with 
  solid thin lines in the central cluster in Figure~\ref{strong-matching-fig-c}), 
  and we claim that no two of those are disjoint. 
  In more detail, and referring to the figure, a strong perfect matching would 
  have to match the bottommost point of the cluster with the central point, but 
  the corresponding triangle intersects all the other 4 empty triangles. 

  Next, suppose that the set contains more than one between-cluster empty 
  triangle assigned to $S_i$.  Then there must be exactly two such triangles, 
  one with a corner in $S_{i-1}$, and one with a corner in $S_{i+1}$.  
  But then all the empty triangles inside $S_i$ are ruled out.
\end{proof}

\para{A lower bound on $\hat\beta(n)$.}

Figure~\ref{blocking-fig} shows how to construct a set of $n$ points with at 
least $(5n-6)/4$ pairwise internally-disjoint empty triangles.  Start with the 
triangle $t$ which has two points on its boundary, then attach to it copies of 
the gadget $\gamma$ stacked one on top of the other; this gadget adds four 
points and five interior-disjoint triangles.

\begin{figure}[t] 
	\centering
	\includegraphics[scale=.75]{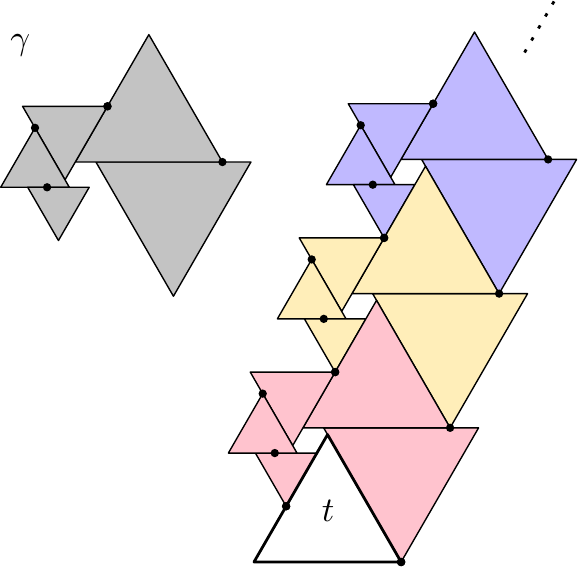} 
	\caption{A set of $n$ points with at least $(5n-6)/4$ pairwise internally disjoint empty triangles.}
	\label{blocking-fig}
\end{figure}

\begin{theorem}
  There are infinitely many~$n$ with $\alpha(n)\le 3n/4$,
  $\mu^*(n)\le 2n/5$, and $\hat\beta(n)\ge(5n-6)/4$.  
\end{theorem}

\section{Additional Properties of \texorpdfstring{$\Theta_6$}{Theta6}-Graphs}
\label{sec:properties}

In this section, we prove some addition structural properties of $\Theta_6$-graphs.  
In particular, we prove bounds on the maximum number of edges, the minimum 
vertex-degree, and the maximum-size of an independent set.
Throughout this section, $P$ denotes a set of $n$ points in the plane.

\para{Edge Density.}

First, we are interested in the density, i.e., the number of edges.  
Clearly, the $\Theta_6$-graph is connected, hence has at least $n-1$ edges, 
and this is achieved for example by points on a vertical line. Morin and Verdonschot~\cite{Morin2014} 
studied the average number of edges of $\Theta_6$-graphs. Together with some other results, they 
showed that the expected number of edges (of the $\Theta_6$-graph of a set of $n$ points, 
chosen randomly, uniformly and independently in a unit square) is $4.186 n\pm O(\sqrt{n\log n})$.   
As for the maximum number, an easy argument shows  that there are at most $5n-11$ edges: 
For any set $P$ of $n\geq 3$ points, the graphs $G^\up(P)$ and $G^\down(P)$ are planar 
and contain at most $3n-6$ edges each. The $n-1$ edges of a minimum spanning tree belong 
to both graphs, so their union contains at most
Also recall that the intersection graph $G^\up(P)\cap G^\down(P)$ is connected, 
and thus has at least $n-1$ edges. Based on these facts, Babu \etal\cite{Babu2014} 
showed that $G^\both$ contains at most $2(3n-6)-(n-1)=5n-11$ edges.   
We can improve this slightly:

\begin{lemma}
Any $\Theta_6$-graph on $n\geq 3$ points has at most $5n-12$ edges.
\end{lemma}
\begin{proof}
  Consider the graphs $G^\up$ and $G^\down$.  If one of them has an outer face 
  that is not a triangle, then it has at most $3n-7$ edges, and re-doing the above 
  analysis gives the bound.  If both $G^\up$ and $G^\down$ have a triangle as 
  outer face, then the vertices on them are necessarily the same three vertices, 
  and the three edges between them belong to both $G^\up$ and $G^\down$ and form a 
  cycle.  Since the minimum spanning tree also belongs to both $G^\up$ and $G^\down$, there are  
  at least $n$ edges common to both graphs, and re-doing the analysis gives the bound.
\end{proof}

It is worth noting that $G^\up$ and $G^\down$ actually cannot both have a 
triangular outer face for $n\geq 4$ since this would contradict Lemma~\ref{lem:f3_4}: 
With $T$, the outer face we would have $f_3^\up=f_3^\down=1$, while $f_{4+}^\up=f_{4+}^\down=0$ 
since there is only one component of $G\setminus T$.

Note that the bound $5n-12$ is tight for $n=3$ if the three points form a triangle.
We do not know whether it is tight for larger $n$.
Babu \etal\cite{Babu2014} found a set of $n$ points whose $\Theta_6$-graph has $(4+1/3)n-13$ edges.  
We can improve on this and show that the factor `5' in the upper bound is tight.

\begin{lemma}
  For any $n\geq 7$, there exists a set of $n$ points whose $\Theta_6$-graph has $5n-17$ edges.
\end{lemma}
\begin{proof}
  See Figure~\ref{fig:many-edges-a}.
  Start with a set~$P$ of $n-6$ points on a vertical line; these have $n-7$ edges 
  between them (black bold).  Add 6 surrounding points $a_1,\dots,a_6$ as in Figure~\ref{augment-fig}. 
  Each of $a_1,a_3,a_4,a_6$ is adjacent to all points of $P$, adding $4n-24$ edges.  
  We are free to move $a_1,\dots,a_6$ (within their respective regions) and can 
  arrange them such that they form an octahedron, adding 12 edges among them 
  (blue dashed). Finally, we have one edge each from $a_2$ and $a_5$ to the 
  topmost/bottommost point of $P$.  
  Hence, in total we have $n-7+4n-24+12+2=5n-17$ edges.
\end{proof}

\begin{figure}[t]
  \centering
	\subcaptionbox{\label{fig:many-edges-a}}{\includegraphics[width=.35\columnwidth]{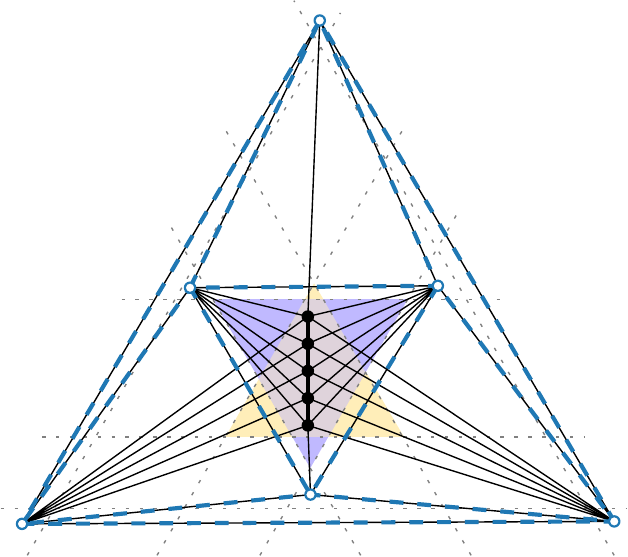}}
\hspace*{-10mm}
	\subcaptionbox{\label{fig:many-edges-b}}{\includegraphics[width=.63\columnwidth,page=2]{high_degree.pdf}}
	\caption{(a) A set of $n=11$ points whose $\Theta_6$-graph has $5n-17=38$ edges.  (b) A set of $26$ points whose $\Theta_6$-graph has minimum degree 7.  Not all edges are shown.}
	\label{fig:many-edges}
	\label{fig:many-edges-new}
	\label{fig:high-degree}
\end{figure}

\para{Vertex-Degrees, Coloring and Independent Sets.}
Since the number of edges of every $\Theta_6$-graph with $n$ vertices is at 
most $5n-12$, its total vertex-degree is at most $10n-24$, so some vertex has degree at most 9.
In particular, therefore $\Theta_6$-graphs are 9-degenerate, which implies that 
they are 10-vertex-colorable (even 10-list-colorable) and have an independent set of size at least $n/10$.  

It remains open whether there are $\Theta_6$-graphs with minimum degree 9 or 
even minimum degree 8, but we can construct one with minimum degree 7;
see Figure~\ref{fig:many-edges-b}.  We construct our graph by starting 
with $K_5\setminus e$ (black bold edges), realized in such a way that each 
vertex $v$ has $7-\deg(v)$ cones that are empty (contain no other point).  
Then we add 6 surrounding points, arranged so that they form an octahedron 
(blue dashed edges). With this, each point of $K_5\setminus e$ obtains another
edge in each of its empty cones and hence has degree 7.  This gives a graph 
where all but two vertices $a,b$ have degree 7 or more.  Taking two copies of 
this graph and placing them such that the copies of $a$ and $b$ become adjacent 
then gives a $\Theta_6$-graph of minimum degree 7.  Note that more edges appear 
between the copies, but the minimum degree remains 7. We summarize in the following theorem.

\begin{theorem}
  Any $\Theta_6$-graph on $n$ points is 10-vertex-colorable and has an
  independent set of size at least~$n/10$. Furthermore, there are
  $\Theta_6$-graphs on $n\ge 11$ points with minimum degree~7.
\end{theorem}

\section{Conclusions and Open Problems}

We have improved the lower bound on the size of a matching in 
any $\Theta_6$-graph on $n$ points to $(3n-8)/7$.  A main open problem is to 
prove the conjecture that any $\Theta_6$-graph has a (near-)perfect matching.

We have shown that this conjecture is equivalent to proving that every 
$\Theta_6$-graph on $n$ points requires at least $n-1$ points to block all its edges.  
More generally, we proved a relationship between the minimum size of maximum 
matchings and the minimum size of blocking sets so that any improvement in the 
lower bound for one of these parameters will also improve the other. 

We have shown that this conjecture is equivalent to proving that every 
$\Theta_6$-graph on $n$ points requires at least $n-1$ points to block all its edges.  
More generally, we proved a relationship between the minimum size of maximum 
matchings and the minimum size of blocking sets so that any improvement in the
lower bound for one of these parameters will also improve the other.

\para{Acknowledgements.}
This work was done by a University of Waterloo problem solving group.  We thank 
the other participants, Alexi Turcotte and Anurag Murty Naredla, for helpful discussions. 

\bibliographystyle{abbrvurl}
\bibliography{Theta6-Matching}

\end{document}